\documentclass[a4paper,UKenglish,cleveref,autoref,thm-restate]{lipics-v2021}

\usepackage{amsmath, amsthm, amsfonts, amssymb}
\usepackage{complexity}
\usepackage{graphicx} 
\usepackage{caption}
\usepackage{subcaption}
\usepackage{xcolor}
\usepackage{nicefrac}
\usepackage{bm}

\graphicspath{figures/}

\newtheorem{problem}{Problem}

\newcommand{\Pol}{\ensuremath{P}}
\newcommand{\PolQ}{\ensuremath{Q}}
\newcommand{\defeq}{\stackrel{{\rm def}}{=}}

\DeclareMathOperator*{\argmin}{arg\,min}

\def\OPT{{\rm OPT}}

\def\H{{\rm H}}

\crefname{problem}{Problem}{Problems}
\crefname{figure}{Figure}{Figures}
\crefname{theorem}{Theorem}{Theorems}
\crefname{lemma}{Lemma}{Lemmas}
\crefname{corollary}{Corollary}{Corollaries}
\crefname{section}{Section}{Sections}
\crefname{appendix}{Appendix}{Appendices}
\crefname{remark}{Remark}{Remarks}
\crefname{claim}{Claim}{Claims}
\crefname{conjecture}{Conjecture}{Conjectures}

\title{Segment Watchman Routes}

\titlerunning{Segment Watchman Routes}

\author{Anna Br\"otzner}{Faculty of Technology and Society, Malm\"o University, Sweden\and\url{https://mau.se/en/persons/anna.brotzner/}}{anna.brotzner@mau.se}{https://orcid.org/0000-0002-2161-6571}{}

\author{Omrit Filtser}{Department of Mathematics and Computer Science, The Open University of Israel, Israel\and\url{https://omrit.filtser.com/}}{omrit.filtser@gmail.com}{https://orcid.org/0000-0002-3978-1428}{}

\author{Bengt~J. Nilsson}{Faculty of Technology and Society, Malm\"o University, Sweden\and\url{https://mau.se/en/persons/bengt.nilsson.ts/}}{bengt.nilsson.TS@mau.se}{https://orcid.org/0000-0002-1342-8618}{}

\author{Christian Rieck}{Institute of Mathematics, University of Kassel, Germany\and\url{https://christianrieck.github.io/}}{christian.rieck@mathematik.uni-kassel.de}{https://orcid.org/0000-0003-0846-5163}{}

\author{Christiane Schmidt}{Department of Science and Technology, Link\"oping University, Sweden\and \url{https://www.itn.liu.se/~chrsc91/}}{christiane.schmidt@liu.se}{https://orcid.org/0000-0003-2548-5756}{}

\authorrunning{A. Brötzner, O. Filtser, B.\,J. Nilsson,  C. Rieck, and C. Schmidt}

\Copyright{Anna Brötzner, Omrit Filtser, Bengt J.~Nilsson, Christian Rieck, and Christiane Schmidt}



\funding{
	Anna Brötzner, Bengt J.~Nilsson, and Christiane Schmidt were supported by grant 2021-03810 (Illuminate: provably good algorithms for guarding problems) from the Swedish Research Council (Vetenskapsr\r{a}det). 
	Anna Brötzner and Bengt J.~Nilsson were also funded by grant 20250562 (Illuminate and Move: provably good algorithms for exploration problems) from the Crafoord Foundation. 
	Omrit~Filtser was supported by the Israel Science Foundation (grant No. 2135/24).
	Christian~Rieck is funded by the Deutsche Forschungsgemeinschaft (DFG, German Research Foundation), grant 522790373.
	}

\nolinenumbers

\hideLIPIcs

\begin{document}

\maketitle

\begin{abstract}
Motivated by applications for robust guarding, we consider a variant of the multiple-watchmen problem that ensures that every point within a polygon $\Pol$ is seen from more than one direction: 
we search for two routes $W_1,W_2$, such that every point $p\in\Pol$ is contained in a segment $\overline{w_1w_2}\subseteq P$ such that $w_1\in W_1$ and $w_2\in W_2$. 
We call such routes \emph{segment watchman routes}.

We show that finding the two routes that are optimal with respect to the min-max criterion is weakly \NP-hard even in simple polygons, and that finding the routes that are optimal with respect to the min-sum criterion is \NP-hard in polygons with holes.
Moreover, we present sufficient conditions for routes to be segment watchman routes, and provide a polynomial-time $2$-approximation under both the min-max criterion and the min-sum criterion, both in simple polygons. 
Finally, we show how to generalize our results for $k$ watchmen.
\end{abstract}

\newpage
\section{Introduction}
In the classical \textsc{Watchman Route Problem}, introduced by Chin and Ntafos~\cite{cn-owr-86,Chin88}, we ask for the shortest route inside a given simple polygon \Pol, such that all points of \Pol\ are visible from at least one point on the route. 
Maybe surprisingly, this problem can be solved in polynomial time~\cite{Tan01a,TanJ17}. 
In this context, a point $p\in \Pol$ \emph{sees} another point $q\in\Pol$ if the line segment $\overline{pq}$ is fully contained in \Pol.

Carlsson, Nilsson, and Ntafos~\cite{cnn-ogcmw-93} raised the $m$-watchmen problem as a natural generalization: 
we are given $m$ watchmen (with or without given starting points) for which we aim to find routes, such that each point in \Pol\ is visible from at least one of the $m$ routes. 
Two common objectives for this problem are to minimize the total length of all $m$ watchman routes (called \emph{min-sum}) and to minimize the length of the longest route assigned to any watchman (called \emph{min-max}). 
In this paper, we mainly focus on the case $m=2$.

In the classical setting with $m$ watchmen, the goal is simply to ensure that each point is observed at least once, without any notion of robustness. 
In practice, however, when robots are used as watchmen, such minimal coverage is often insufficient: 
robots in multi-robot systems may fail or encounter obstacles, particularly in remote or hazardous regions, and observing a point from multiple perspectives can significantly enhance the reliability and quality of monitoring. 
To bridge this gap between theory and application, we focus on routes that are robust and effective in real-world scenarios. 
Specifically, in this paper, we aim to improve coverage quality by ensuring that each point is observed from multiple directions. 
This is a novel variant of the $m$-watchmen problem inspired by similar coverage requirements that have been studied for static guards~\cite{ehm-aatol-05,se-tg-03} (see related work below). 
Because each point is seen by each watchman, we will be able to perceive the complete environment even in case of all but one of the robots failing.
This redundancy provides robustness (see, e.g.,~\cite{fdm-mrrac-13}). 
However, without failures, we cover even another robustness aspect: for certain types of laser scanners, the quality of the attained point clouds depends on the incidence angle (i.e., the angle between incoming laser beam and the local surface normal)~\cite{bm-etcsg-18,slmt-sgifq-11,tc-ciade-16}. 
Hence, to be robust against impaired perception with a multi-robot system (our watchmen) equipped with such lasers, we aim to see a point from different directions:
for $m=2$, we require that each point $p\in\Pol$ is seen from two points on the two watchmen's routes that define a segment on which $p$ must lie, that is, $p$ must be seen from opposite sides. 

\subparagraph*{Problem Definition.}
Let \Pol\ be a \emph{simple polygonal domain}.
A point $p\in\Pol$ is \emph{segment-guarded} by two points $w_1,w_2 \in\Pol$ if each of $w_1,w_2$ sees $p$ and the segment $\overline{w_1w_2}$ contains $p$.
We say that two routes $W_1,W_2$ in \Pol\ are \emph{segment watchman routes} for \Pol\ if for every point $p\in\Pol$ there exist two points $w_1\in W_1$, $w_2\in W_2$ such that $p$ is segment-guarded by $w_1$ and $w_2$.

\medskip
We always assume that a route is closed, that is, it starts and ends at the same point. 
For a route $W$ in \Pol, let $|W|$ represent its length. 
We focus on the following two problems:

\begin{problem}[Min-Max Segment Watchmen Route Problem]\label{prob:mm-seg}
    Given a polygonal domain \Pol, find segment watchman routes $W_1,W_2$ such that $\max \{|W_1|, |W_2|\}$ is minimized.
\end{problem}

\begin{problem}[Min-Sum Segment Watchmen Route Problem]\label{prob:ms-seg}
    Given a polygonal domain \Pol, find segment watchman routes $W_1,W_2$ such that $|W_1| + |W_2|$ is minimized.
\end{problem}

In \Cref{fig:non-trivial-min-max-example}, we illustrate a collection of example polygons along with their corresponding segment watchman routes. 
Note that the routes may overlap.
Also note that segment-guarding a point~$p$ by two watchmen on routes $W_1$ and $W_2$ does not require them to occupy positions $w_1 \in W_1$ and $w_2 \in W_2$ simultaneously.

\begin{figure}[htb]
    \captionsetup[subfigure]{justification=centering}%
    \begin{subfigure}{0.33\columnwidth}
    \centering
        \includegraphics[scale=0.8,page=1]{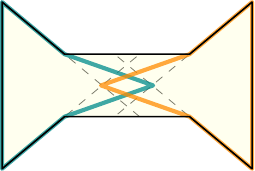}
        \caption{}
        \label{fig:non-trivial-min-max-example-a}
    \end{subfigure}%
    \begin{subfigure}{0.33\columnwidth}
    \centering
        \includegraphics[scale=0.8,page=2]{figures/non-trivial-segment-routes-new.pdf}
        \caption{}
        \label{fig:non-trivial-min-max-example-b}
    \end{subfigure}%
        \begin{subfigure}{0.33\columnwidth}
        \centering
        \includegraphics[scale=0.8,page=3]{figures/non-trivial-segment-routes-new.pdf}
        \caption{}
        \label{fig:non-trivial-min-max-example-c}
    \end{subfigure}%
    \caption{Min-max segment watchman routes (a) may or (b) may not need to overlap. 
    Moreover, (c) segment watchman routes do not ``sweep'' the polygon, and therefore are not required to see a point at the same time.
    (Note again that routes are closed, hence, if we depict a segment only, this means that the route goes back and forth along the segment.)}
    \label{fig:non-trivial-min-max-example}
\end{figure}

\subparagraph*{Our Contribution.}
We introduce the \textsc{Segment Watchman Routes Problem}, and study it in polygons with and without holes, under both the min-sum and min-max objectives.
We~provide sufficient conditions for $2$-watchmen routes to be segment watchman routes.

On the algorithmic side, we present a $2$-approximation algorithm for both the min-max and min-sum objectives. 
Moreover, we establish \NP-hardness for simple polygons under the \mbox{min-max} measure, and for polygons with holes under the min-sum objective.

Finally, we generalize the concept of being segment-guarded to being \emph{$k$-hull-guarded}: 
A point $p\in\Pol$ is \emph{$k$-hull-guarded} by $k$ points $w_1,w_2,\dots,w_k\in\Pol$ if the convex hull of  $w_1,w_2,\dots,w_k$ contains $p$, and each of $w_1,w_2,\dots,w_k$ sees $p$. 
Similarly to segment watchman routes, we say that $k$ routes $W_1,W_2,\dots,W_k$ in \Pol\ are \emph{$k$-hull watchman routes} for \Pol\ if for every point $p\in\Pol$ there exist $k$ points $w_i\in W_i$ for $i\in\{1,\dots,k\}$ such that~$p$ is $k$-hull-guarded by $w_1,w_2,\dots,w_k$.
We show that most of our results for segment watchman routes carry over to $k$-hull watchman routes. In particular, the strategy of our approximation algorithm yields a $k$-approximation for the $k$-hull watchman routes problem under the min-max criterion, and a $2$-approximation under the min-sum criterion. 
Moreover, the \NP-hardness results also carry over to the respective setting.

\subparagraph*{Related Work.} 
Carlsson, Nilsson, and Ntafos~\cite{cnn-ogcmw-93} showed that the $m$-watchmen problem is \NP-hard in simple polygons and provided a polynomial time algorithm for histograms. 
Polynomial time algorithms for different polygon classes, using either the min-sum or the min-max objective, have also been presented~\cite{bagheri23journal, mw-wrmg-91, ns-smwrh-92, nw-wrsp-90}. 
Nilsson and Packer~\cite{nilssonpacker2024} proposed a 5.969-approximation algorithm to compute min-max $2$-watchman routes in~simple~polygons without given starting points (also called \emph{floating}), and obtained a factor of~$\approx 6.922$ for the case with given starting points (also called \emph{anchored}).
Very recently, Brötzner, Nilsson, and Schmidt~\cite{bns-iatwr-26} improved these approximation factors to $3+\pi/2\approx 4.571$ for the floating, and to~$2+\pi/2 \approx 3.571$ for the anchored version. 

The robustness requirement we employ for watchman routes in this paper is closely related to the problems of two-sensor visibility and triangle guarding for stationary guards introduced by Efrat, Har-Peled, and Mitchell~\cite{ehm-aatol-05} and Smith and Evans~\cite{se-tg-03}, respectively. 
Both considered two polygons $\Pol$ and $\PolQ$ with $\PolQ\subseteq\Pol$, where the subpolygon $\PolQ$ should be guarded by guards placed in~\Pol\ (assuming that $\PolQ$'s boundary is transparent). 
In~\cite{ehm-aatol-05}, a point $p\in \PolQ$ is \emph{2-guarded at angle $\alpha$} by two guards $g_1,g_2$ if $\angle g_1 p g_2 \in[\alpha, \pi-\alpha]$ and both guards see~$p$. 
Smith and Evans defined a point $p\in \PolQ$ to be \emph{triangle-guarded} by $g_1, g_2, g_3$ if $p$ is seen by each of the three guards and is contained in the triangle spanned by them. 
Another variant of robust guarding has recently been established by Das, Filtser, Katz, and Mitchell~\cite{dfkm-rgp-24,DasFKM24}; and a variant of robustness for a single watchman by Langetepe, Nilsson, and Packer~\cite{lnp-dstpd-17}.

Huynh, Mitchell, and Polishchuk~\cite{HuynhMP25} consider a problem that is similar to our notion of segment watchman routes. 
They study a version of monitoring a polygon via line segments: sweeping a polygonal domain with variable-length segments whose endpoints are two mobile agents. 
A point is seen if at some point in time, it lies on a segment between the current positions of the two agents. 
This differs from our definition of being segment-seen by the temporal constraint, as depicted in \Cref{fig:non-trivial-min-max-example-c}.

\section{Preliminaries and Key Lemmas}\label{sec:prel}
Let \Pol\ be a simple polygon with $n$ vertices. Denote by $\partial \Pol$ the boundary of \Pol,\, and for two points $a,b\in \partial \Pol$ denote by $\partial\Pol[a,b]$ the boundary of \Pol\ between $a$ and $b$ in counterclockwise order. For two points $a,b\in \Pol$, denote by $\pi(a,b)$ the geodesic path between $a$ and $b$ in \Pol\ (i.e., the shortest path in \Pol\ connecting $a$ and $b$).
We assume that \Pol\ does not contain vertices with an internal angle of exactly $180^\circ$, i.e., no three consecutive vertices are on the same line. If~\Pol\ does contain such a vertex, we can simply remove it. 

A set $H\subseteq\Pol$ is \emph{geodesically convex} if for any two points $p,q\in H$, the geodesic $\pi(p,q)$ lies in $H$.
The \emph{relative convex hull} $\H(X)$ of a set $X$ of points in a simple polygon \Pol\ is the smallest geodesically convex set within \Pol\ that contains $X$. 
The relative convex hull~$\H(X)$ forms a weakly simple polygon, that is, a polygon whose boundary does not properly self-intersect, although vertices may be repeated. Its boundary is composed of geodesic paths that connect pairs of points in $X$. 
When $X$ is a route, we slightly abuse the notation and refer to the boundary of $\H(X)$ as the relative convex hull of $X$.

A route $W$ in \Pol\ is called a \emph{watchman route} if every point in \Pol\ is visible from some point on $W$\!. Recall that the length of $W$ is denoted by~$|W|$. We say that $W$ is \emph{simple} if it does not cross itself, and that $W$ is \emph{relatively convex} if it coincides with its relative convex hull $\H(W)$.

Let $W_1$ and $W_2$ be segment watchman routes for \Pol\!. By definition, the following~holds:

\begin{observation}\label{obs:watchmen-routes}
    Each of the routes $W_1$ and $W_2$ is a watchman route for~\Pol.
\end{observation}
Note that \Cref{obs:watchmen-routes} also holds for $k$-hull watchman routes.

The \emph{visibility polygon} of a point $p\in \Pol$ is the set of points in \Pol\ that are visible from~$p$.
The following observation is well known and largely considered folklore; however, we include it for completeness.
\begin{restatable}{observation}{clmWatchmanRoutes}\label{clm:cond2}
    Let $W$ be a route in \Pol\!, then $W$ is a watchman route if and only if it visits the visibility polygon of every convex vertex of \Pol\!.
\end{restatable}

\begin{claimproof}
	One direction is by definition: if $W$ is a watchman route then it must see all the convex vertices, and thus it visits the visibility polygon of every convex vertex.
	
	For the second direction, let $W$ be a route that visits the visibility polygon of every convex vertex, and assume by contradiction that there is a point $p\in\Pol$ that is not seen by~$W$, that is, no point of $W$ lies in~$p$'s visibility polygon. 
	Hence, $W$ is fully contained in one of the \emph{pockets} $\Pol'$ of $p$'s visibility polygon (i.e., a subpolygon of \Pol\ in which no point is visible from $p$, as depicted in~\Cref{fig:pocket}).
	
	\begin{figure}[htb]
		\centering
		\includegraphics[page=3]{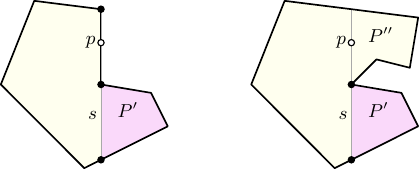}
		\caption{$\ell\setminus s$ is either an edge of \Pol\ with a convex endpoint (left), or it splits \Pol\ into at least two subpolygons, one of which also lies to the right of $\ell$ (right).}
		\label{fig:pocket}
	\end{figure}
	
	Extend the pocket's \emph{window}~$s$ (the line segment that separates $\Pol'$ and $\Pol\setminus\Pol'$) into a maximal line segment $\ell$ contained in \Pol. 
	Without loss of generality, let $\ell$ be a vertical line segment with~$\Pol'$ to its right. 
	As $p\in\ell$, $\ell\setminus s$ is either a polygonal edge with a convex endpoint not seen by $W$, or it splits~\Pol\ into at least two subpolygons. 
	In the latter case, at least one of the subpolygons, call it $\Pol''$, also lies to the right of $\ell$.
	$W$ cannot see any convex vertex in~$\Pol''$, yielding a contradiction.
\end{claimproof}

We next establish sufficient conditions ensuring that two routes form segment watchman routes.
We subsequently extend the argument to the general $k$-hull setting in~\cref{sec:kgon-gen}.

\begin{lemma}[The Conditions Lemma for Segment Watchman Routes]\label{lem:conditions}
    Two routes $W_1$ and $W_2$ are segment watchman routes for \Pol\ if the following conditions hold:
    \begin{enumerate}[(i)]
        \item Every convex vertex is visited by one of $W_1$ or~$W_2$.\label{cond1}
        \item Both $W_1$ and $W_2$ visit the visibility polygon of each convex vertex.\label{cond2}
        \item Both $W_1$ and $W_2$ are relatively convex.\label{cond3}
    \end{enumerate}
\end{lemma}

\begin{proof}
By \Cref{clm:cond2}, Condition (\ref{cond2}) implies that $W_1$ and $W_2$ are watchman routes. 
Consider a point $p\in\Pol$\!.
Since both $W_1$ and $W_2$ are watchman routes, there exists at least one point on $W_1$ and at least one point on $W_2$ that $p$ sees. 
Consider the two wedges defined by the angles from which $p$ is viewing $W_1$ and $W_2$, as visualized in~\Cref{fig:conditions}: let $F_1$ be the maximal wedge bounded by two rays starting at $p$, such that for every ray $\rho$ in $F_1$ there is a point $w\in W_1$ in this direction that $p$ sees. 
Note that because \Pol\ is simple, $F_1$ is a single wedge. 
The wedge $F_2$ is defined analogously for $W_2$.

\begin{figure}[htb]
    \centering
    \includegraphics[page=1]{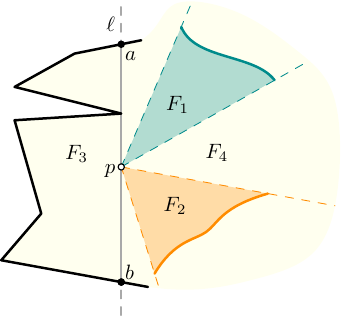}
    \caption{The wedges $F_1$ and $F_2$ define the angles from which $p$ is viewing $W_1$ and $W_2$, respectively. If $F_3$ or $F_4$ is larger than $180^\circ$, then there is a convex vertex on the left side of $\ell$ which is not visited.}
    \label{fig:conditions}
\end{figure}

Consider the two wedges $F_1$ and $F_2$. If $p$ lies on or within one relatively convex route, the corresponding wedge covers $360^\circ$. 
If $p$ does not lie on or within one route, 
the corresponding wedge covers less than $180^\circ$ because by Condition (\ref{cond3}) both routes are relatively convex. 
If at least one of $F_1$ and $F_2$ covers $360^\circ$ around~$p$, then $p$ is segment-guarded: assume that $F_2$ covers $360^\circ$ around~$p$, and let $w_1$ be a point on $W_1$ that sees $p$. 
Then the ray from~$w_1$ in the direction of~$p$ intersects $W_2$ at point $w_2$ that sees $p$, and thus $p$ is segment-guarded by $\overline{w_1w_2}$.
Hence, assume that neither $F_1$ nor $F_2$ covers $360^\circ$ around $p$. 
Let $F_3$ (and, if it exists, $F_4$) be the maximal wedge(s) bounded by two rays originating at $p$, such that for every ray $\rho$ within $F_3$ (and $F_4$), there is no point $w \in W_1$ or $w \in W_2$ along $\rho$ that is visible from $p$.
The plane around~$p$ can then be partitioned into up to four wedges, depending on whether $F_1$ and~$F_2$ intersect. In the non-intersecting case, the wedges are $F_1$, $F_2$, $F_3$, and $F_4$. 
If~$F_1$ and~$F_2$ overlap, the wedges are $F_1$, $F_2$, $F_3$, and a single wedge corresponding to the intersection of $F_1$ and $F_2$.

We now argue that neither $F_3$ nor $F_4$ can cover more than $180^\circ$. This implies that one of the rays through $p$ that define $F_1$ must meet a point on $W_2$ in the opposite direction from $p$, and therefore $p$ is segment-guarded and the claim follows.

Without loss of generality, assume that $F_3$ covers more than $180^\circ$ (as shown in~\cref{fig:conditions}).
Consider a line $\ell$ through $p$ in $F_3$ that does not contain an edge of the boundary of \Pol, and assume that $\ell$ is a vertical line and that both $F_1$ and $F_2$ are on the right side of $\ell$. 
Let $\overline{ab}$ be the maximal line segment on $\ell$ that is contained in~\Pol\!. 
Then $\overline{ab}$ splits \Pol\ into at least two subpolygons, and at least one of them, $\Pol'\!\!$, is on the left side of $\overline{ab}$. Because \Pol\ is simple and both $W_1$ and $W_2$ do not cross $\overline{ab}$, there are no points of $W_1$ and $W_2$ in~$\Pol'\!\!$.
However, $\Pol'$ must contain a convex vertex $v$. This yields a contradiction, as by Condition~(\ref{cond1}), $v$ needs to be visited by at least one of the watchman~routes.
\end{proof}

\begin{figure}[htb]
    \captionsetup[subfigure]{justification=centering}%
    \centering
    \begin{subfigure}{0.25\columnwidth}%
    \centering
        \includegraphics[page=4]{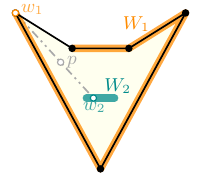}%
        \caption{}%
        \label{fig:cond-vs-opt-a}
    \end{subfigure}%
        \begin{subfigure}{0.25\columnwidth}%
    \centering
        \includegraphics[page=5]{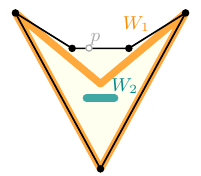}%
        \caption{}%
        \label{fig:cond-vs-opt-b}
    \end{subfigure}%
    \begin{subfigure}{0.25\columnwidth}%
    \centering
        \includegraphics[page=6]{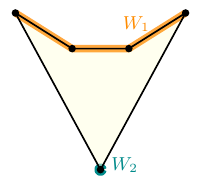}%
        \caption{}%
        \label{fig:cond-vs-opt-c}
    \end{subfigure}%
        \begin{subfigure}{0.25\columnwidth}%
    \centering
        \includegraphics[page=7]{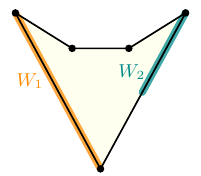}%
        \caption{}%
        \label{fig:cond-vs-opt-d}
    \end{subfigure}%
    \caption{ (a) $W_1$ and $W_2$ are segment watchman routes (e.g., $p$ lies on $\overline{w_1 w_2}$), but do not fulfill the conditions of~\cref{lem:conditions}. 
    They are not optimal: $W_1$'s relative convex hull $H(W_1)$ (the boundary of the polygon) is shorter than $W_1$, and $H(W_1)$ together with $W_2$ are segment watchman routes.
    (b)~$W_1$ and $W_2$ do not fulfill Condition (iii) of~\cref{lem:conditions} and they are not segment watchman routes, since $p$ is not properly guarded. 
    (c) and (d) depict segment watchman routes that fulfill all the conditions of~\cref{lem:conditions}.}%
    \label{fig:cond-vs-opt}
\end{figure}

In~\cref{lem:conditions}, the conditions imply that $W_1$ and $W_2$ are segment watchman routes. However, there exist segment watchman routes that do not fulfill these conditions, see~\cref{fig:cond-vs-opt}.
On~the~other hand, there always exists an optimal solution for the segment watchman route problem in any simple polygon with respect to the min-max and min-sum criterion for which Conditions (\ref{cond1}), (\ref{cond2}), and (\ref{cond3}) hold.

\begin{lemma}\label{obs:optimal-segment-1-2}
Let \Pol\ be a simple polygon. 
If two routes $W_1$ and $W_2$ are segment watchman routes for \Pol, then Conditions (\ref{cond1}) and (\ref{cond2}) from~\cref{lem:conditions} hold.
\end{lemma}
\begin{proof}
    First, notice that by \Cref{clm:cond2}, Condition (\ref{cond2}) is a necessary condition, that is, if $W_1$ and~$W_2$ are segment watchman routes then Condition~(\ref{cond2}) holds. 
    
    Next, we claim that Condition~(\ref{cond1}) is also a necessary condition, that is, we show that if $W_1$ and $W_2$ are segment watchman routes, then every convex vertex of \Pol\ must be visited by one of $W_1$ or $W_2$. Let $v$ be a convex vertex of \Pol. 
    Then $v$ lies on a line segment $\overline{w_1w_2}$ with $w_1\in W_1$ and $w_2\in W_2$, and the segments $\overline{vw_1}$ and $\overline{vw_2}$ are contained in \Pol. As the interior angle at $v$ is strictly smaller than $180^\circ$\!, any line segment in \Pol\ that contains $v$ has $v$ as one of its endpoints, and therefore $v$ is visited by either $W_1$ or $W_2$.
\end{proof}

\begin{lemma}\label{obs:optimal-segment}
Let \Pol\ be a simple polygon. 
There always exists an optimal solution for the segment watchman route problem in \Pol\ with respect to the min-max and min-sum criterion for which Conditions~(\ref{cond1}), (\ref{cond2}), and (\ref{cond3}) from~\cref{lem:conditions} hold.
\end{lemma}
\begin{proof}
    Let $W_1$ and $W_2$ be an optimal solution for the segment watchman route problem in~\Pol\ (with respect to either the min-max or the min-sum criterion). 
    Conditions (\ref{cond1}) and (\ref{cond2}) from~\cref{lem:conditions} hold by \Cref{obs:optimal-segment-1-2}. 

    We show that $\H(W_1)$ and $\H(W_2)$ are segment watchman routes. Since $|\H(W_i)|\le W_i$, this would prove our claim. 
    Clearly, since Condition~(\ref{cond1}) holds for $W_1$ and $W_2$, it also holds for $\H(W_1)$ and $\H(W_2)$. 
    Condition~(\ref{cond3}) holds by definition. 
    We therefore only need to prove that Condition~(\ref{cond2}) still holds.
    Let $v$ be a convex vertex of \Pol, and assume that $W_i$ visits the visibility polygon of~$v$. Let $w$ be a point on $W_i$ that sees $v$. 
    The segment $\overline{vw}$ is contained in~$\Pol$, and because~$\H(W_i)$ contains $W_i$, its boundary must cross $\overline{vw}$ at a point $w'$. We conclude that $\H(W_i)$ also visits the visibility polygon of $v$, and thus all conditions from \Cref{lem:conditions} hold, and both $\H(W_1)$ and $\H(W_2)$ are segment watchman routes.
\end{proof}

\section{Approximation Algorithms for Segment Watchman Routes}\label{sec:appx-segm}
We begin, as in the previous section, by considering segment watchman routes. 
We generalize the ideas developed in this section in~\Cref{sec:kgon-gen}.

Let $P$ be a simple polygon with $t$ convex vertices. 
We enumerate the convex vertices in counterclockwise order $V_C=\{v_0,\dots, v_{t-1}\}$, with $v_0$ chosen arbitrarily. In the following, we assume without loss of generality that indices are counted modulo $t$.

\begin{restatable}{lemma}{obsRelHullShortest}\label{obs:relative_hull_is_shortest}
    Let $X\subseteq V_C$ be a subset of the convex vertices of \Pol. 
    Then $\H(X)$ is a shortest route in \Pol\ that visits $X$.
\end{restatable}

To prove \Cref{obs:relative_hull_is_shortest}, we first need the following observation, which is obtained by uncrossing the geodesic paths that are connecting the convex vertices, as visualized in~\Cref{fig:route-cross}.
\begin{observation}\label{obs:simple_route}
	Let $W$ be a shortest route that visits a subset $X\subseteq V_C$ of convex vertices of \Pol. Then there always exists a route $W'$ with the same length as $W$, that visits the vertices in $X$ according to their order along $\partial P$.
\end{observation}
\begin{figure}[htb]
	\centering
	\includegraphics[scale=0.8]{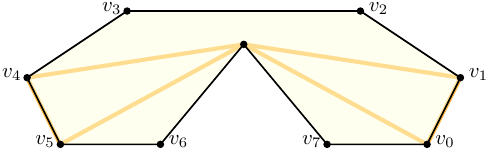}
	\caption{A simple polygon with 8 convex vertices. The route $v_5\rightarrow v_1\rightarrow v_0\rightarrow v_4$ can be uncrossed to obtain the simple path $v_5\rightarrow v_0\rightarrow v_1\rightarrow v_4$.}
	\label{fig:route-cross}
\end{figure}

\begin{proof}[Proof of~\cref{obs:relative_hull_is_shortest}]
	First notice that the route $\H(X)$ must visit all the vertices in $X$, because they are on $\partial\Pol$ and \Pol\ is a simple polygon. 
	
	Let $W$ be a shortest route that visits all the vertices in $X$. 
	By \Cref{obs:simple_route}, we may assume that it visits the vertices in $X$ according to their order along $\partial\Pol$. 
	Consider two vertices $v_a,v_b\in X$ that are consecutive in $X$ according to the order along $\partial\Pol$. 
	Then the path connecting them in $W$ must be the geodesic path $\pi(v_a,v_b)$, because this is the shortest path between them. 
	Applying this argument to any pair of consecutive vertices in $X$, we get that $W$ is exactly $\H(X)$.
\end{proof}

\begin{restatable}{lemma}{obsRCHValid}\label{obs:RCH-is-valid}
    Let $C$ be a route that visits a subset $X\subset V_C$ of the convex vertices of~\Pol, and $D$ be a route that sees another subset $Y\subset V_C$. Then the route $W=\H(C\cup D)$ visits the convex vertices in $X$ and sees the convex vertices in $Y$.
\end{restatable}

\begin{proof}
	Because $X\subset W$ and \Pol\ is simple, $W$ has to visit the convex vertices in $X$. Consider a vertex $v\in Y$. If $v$ is in the convex hull $\H(C \cup D)$, then again because \Pol\ is simple, $v$ must be on $W$. Otherwise, let $w$ be a point on $D$ that sees $v$. Then because $\H(C \cup D)$ contains the route $D$, and $v$ is not in $\H(C \cup D)$, the route $W$ must cross the segment $\overline{vw}$, and therefore it sees $v$ from the crossing point.
\end{proof}

Let $v_i,v_j\in V_C$ be two different convex vertices and let $C_{ij}=\H(\{v_i,v_{i+1},\ldots,v_{j-1}\})$ be a shortest route that visits the convex vertices $v_i,v_{i+1},\ldots,v_{j-1}$ as guaranteed by \Cref{obs:relative_hull_is_shortest}. The route ${C_{ji}=\H(\{v_j,v_{j+1}\ldots,v_{i-1}\})}$ is then a shortest route that visits the convex vertices $v_j,v_{j+1}\ldots,v_{i-1}$. \cref{fig:C-and-D-and-W,fig:C-and-D-and-W-2} exemplify the routes defined in this subsection. 
Denote $C_P=\partial \Pol$, and note that it is the shortest route that visits all convex vertices of~\Pol\!.

Let $D_{ij}$ be the shortest route that 
starts and ends at $v_i$, and that 
sees all the convex vertices $v_j,\ldots,v_{i-1}$. The route $D_{ji}$ is then the shortest route that 
starts and ends at $v_j$, and that 
sees all the convex vertices $v_i,\ldots,v_{j-1}$. 
Let $D_{\Pol}$ be the shortest floating watchman route in {\Pol} (that is, the shortest watchman route without a given starting point).

\begin{figure}[htb]
    \centering
    \captionsetup[subfigure]{justification=centering}%
    \centering
    \begin{subfigure}{0.33\columnwidth}%
    \centering
        \includegraphics[page=2,scale=0.85]{figures/2-approx-relabelled.pdf}%
        \caption{}%
        \label{fig:two-approx-a}
    \end{subfigure}%
        \begin{subfigure}{0.33\columnwidth}%
    \centering
        \includegraphics[page=1,scale=0.85]{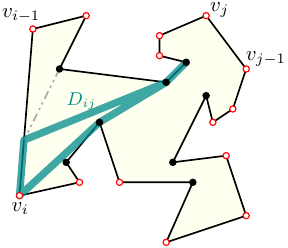}%
        \caption{}%
        \label{fig:two-approx-b}
    \end{subfigure}%
    \begin{subfigure}{0.33\columnwidth}%
    \centering
        \includegraphics[page=3,scale=0.85]{figures/2-approx-relabelled.pdf}%
        \caption{}%
        \label{fig:two-approx-c}
    \end{subfigure}%
    \caption{Examples of the routes used in~\cref{sec:appx-segm}:
    (a) $C_{ij}=\H(\{v_i,v_{i+1},\ldots,v_{j-1}\})$ denotes a shortest route that visits the convex vertices $v_i,v_{i+1},\ldots,v_{j-1}$. (b) $D_{ij}$ is a shortest route that starts and ends at $v_i$ and that sees all the convex vertices $v_j,\ldots, v_{i-1}$. (c) depicts $W_{ij}= {\rm H}(C_{ij} \cup D_{ij})$.
   }
    \label{fig:C-and-D-and-W}
\end{figure}

\begin{figure}[htb]
    \centering
    \captionsetup[subfigure]{justification=centering}%
    \centering
    \begin{subfigure}{0.33\columnwidth}%
    \centering
        \includegraphics[page=5,scale=0.85]{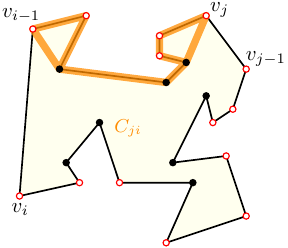}%
        \caption{}%
        \label{fig:two-approx-2-a}
    \end{subfigure}%
        \begin{subfigure}{0.33\columnwidth}%
    \centering
        \includegraphics[page=4,scale=0.85]{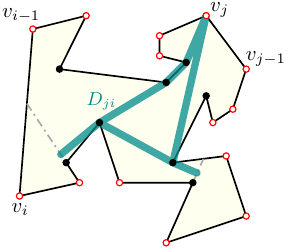}%
        \caption{}%
        \label{fig:two-approx-2-b}
    \end{subfigure}%
    \begin{subfigure}{0.33\columnwidth}%
    \centering
        \includegraphics[page=6,scale=0.85]{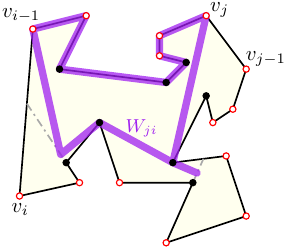}%
        \caption{}%
        \label{fig:two-approx-2-c}
    \end{subfigure}%
    \caption{Complementary routes to those shown in~\cref{fig:C-and-D-and-W}: (a) $C_{ji}$, (b) $D_{ji}$, and (c) $W_{ji}$.
   }
    \label{fig:C-and-D-and-W-2}
\end{figure}

\pagebreak
We define $W_{ij}\defeq {\rm H}(C_{ij} \cup D_{ij})$, i.e., $W_{ij}$ is the relative convex hull of the route obtained by connecting 
the two routes $C_{ij}$ and $D_{ij}$ at $v_i$. 
Then, we construct our approximate solution by choosing the pair 
\[
(W_{1},W_{2})
=
\argmin_{i\neq j} \big\{ \max\{\vert W_{ij}\vert,\vert W_{ji}\vert\}, \max\{\vert C_{\Pol}\vert,\vert D_{\Pol}\vert\} \big\}.
\]

By \cref{lem:conditions,obs:RCH-is-valid}, the pair $(W_{1},W_{2})$ constitutes a feasible solution to the segment watchman routes problem. 

For two routes $W_1$ and $W_2$, the subsequent~\cref{lem:twoapprox,lem:twoapprox-min-sum} establish upper bounds relative to the length of an optimal solution for \Pol\ on both the longer of $W_1$ and $W_2$, and on their sum.
We start with the min-max objective. 
For this, let $\OPT(\Pol)$ denote the length of an optimal solution for a polygon \Pol\ under the min-max criterion.

\begin{lemma}\label{lem:twoapprox}
$\max\big\{\vert W_{1}\vert,\vert W_{2}\vert\big\} \leq 2 \cdot {\rm\OPT}(\Pol)$.
\end{lemma}
\begin{proof}
Let $W_1^*$ and $W_2^*$ be two segment watchman routes with ${\max\big\{\vert W_1^*\vert,\vert W_2^*\vert\big\}=\OPT(\Pol)}$. 
Without loss of generality, we may assume that $W_1^*$ and $W_2^*$ are as short as possible.

If $W_1^*$ or $W_2^*$ visits all convex vertices of \Pol, then $(C_{\Pol},D_{\Pol})$ is an optimal solution to the problem and the theorem therefore holds. Hence, for the remainder of this proof, we assume that $W_1^*$ visits some fixed convex vertex $v_i$ and $W_2^*$ visits a different fixed convex vertex~$v_j$.

Since $W_1^*$ visits $v_{i}$ and it either sees or visits the convex vertices $v_{j},\ldots,v_{i-1}$ by construction, we have that $\vert D_{ij}\vert \leq\vert W_1^*\vert $. 
Similarly, $W_2^*$ visits $v_{j}$ and it either sees or visits the convex vertices $v_{i},\ldots,v_{j-1}$, yielding $\vert D_{ji}\vert \leq\vert W_2^*\vert $.
We distinguish the following cases.

\medskip
\begin{description}
\item[$W_1^*$ and $W_2^*$ do not intersect.]

Because $W_1^*$ and $W_2^*$ do not intersect, the two convex vertices $v_{i}$ and $v_{j}$ can be chosen so that $W_1^*$ visits $v_{i},\ldots,v_{j-1}$ by increasing index (modulo~$t$) and sees the remaining ones, whereas $W_2^*$ visits $v_{j},\ldots,v_{i-1}$ and sees the remaining ones. 
From this, it follows that $\vert C_{ij}\vert \leq\vert W_1^*\vert $ and $\vert C_{ji}\vert \leq\vert W_2^*\vert $. 
We obtain that
\begin{equation}\label{eq:2approx1}
    \begin{split}
\max\big\{\vert W_{1}\vert ,\vert W_{2}\vert \big\}
&\leq
\max\big\{\vert {\rm H}(C_{ij}\cup D_{ij})\vert ,\vert {\rm H}(C_{ji}\cup D_{ji})\vert \big\}\\
&\leq
\max\big\{2\vert W_1^*\vert ,2\vert W_2^*\vert \big\}
=
2\cdot\max\big\{\vert W_1^*\vert ,\vert W_2^*\vert \big\}
=
2\cdot\OPT(\Pol).
\end{split}
\end{equation}

\item[$W_1^*$ and $W_2^*$ intersect.]
Because $W_1^*$ and $W_2^*$ intersect and together visit all the vertices, we have $\vert C_{\Pol}\vert \leq \vert W_1^*\cup W_2^*\vert =\vert W_1^*\vert +\vert W_2^*\vert $ and $\vert D_{\Pol}\vert \leq \min\{\vert W_1^*\vert ,\vert W_2^*\vert \}$, as both $W_1^*$ and~$W_2^*$ are watchman routes. 
We obtain that
\begin{equation}\label{eq:2approx2}
\begin{split}
\max\big\{\vert W_{1}\vert ,\vert W_{2}\vert \big\}
&\leq
\max\big\{\vert C_{\Pol}\vert ,\vert D_{\Pol}\vert \big\}
\leq
\max\big\{\vert W_1^*\vert +\vert W_2^*\vert ,\min\{\vert W_1^*\vert ,\vert W_2^*\vert \}\big\}\\
&\leq
2\cdot\max\big\{\vert W_1^*\vert ,\vert W_2^*\vert \big\}
=
2\cdot \OPT(\Pol).
\end{split}
\end{equation}
\end{description}

Since the case distinction is exhaustive, this concludes the proof.
\end{proof}

Note that in fact, for the min-max criterion we may also let $W_{2}=C_{\Pol}$ to avoid computing a floating shortest watchman route.
We proceed with the min-sum objective. 
For this, let~$\OPT^\Sigma(\Pol)$ denote the size of an optimal solution for~\Pol\ under the min-sum criterion.
\begin{lemma}\label{lem:twoapprox-min-sum}
$\vert W_{1}\vert+\vert W_{2}\vert \leq 2 \cdot {\OPT^\Sigma}(\Pol)$.
\end{lemma}
\begin{proof}
    The proof proceeds along the same lines as that of \Cref{lem:twoapprox}. 
    In particular, it suffices to replace \Cref{eq:2approx1} and \Cref{eq:2approx2} with the following two equations, respectively.
    
    \begin{description}
        \item[$W_1^*$ and $W_2^*$ do not intersect.]
            \[
                \vert W_{1}\vert+\vert W_{2}\vert \leq \vert {\rm H}(C_{ij}\cup D_{ij})\vert +\vert {\rm H}(C_{ji}\cup D_{ji})\vert \le 2\cdot(\vert W_1^*\vert +\vert W_2^*\vert) = 2\cdot\OPT^\Sigma(\Pol).
            \]
        \item[$W_1^*$ and $W_2^*$ intersect.]
            \[
                \vert W_{1}\vert +\vert W_{2}\vert \leq \vert C_{\Pol}\vert + \vert D_{\Pol}\vert \leq \vert W_1^*\vert +\vert W_2^*\vert +\min\{\vert W_1^*\vert ,\vert W_2^*\vert \} \leq \nicefrac{3}{2}\cdot \OPT^\Sigma(\Pol).\qedhere
            \]
    \end{description}
\end{proof}

Using~\cref{lem:twoapprox,lem:twoapprox-min-sum}, we obtain the following theorem.
\begin{theorem}\label{thm:twoapprox}
    Given a simple polygon \Pol\ with $n$ vertices, out of which $t$ are convex, one can compute in $\mathcal{O}(n^4+t^2n^3)$ time a 2-approximation for the segment watchman routes problem under both criteria.
\end{theorem}
\begin{proof}
    Feasibility follows directly from \Cref{lem:conditions}, while the approximation guarantees are established by \Cref{lem:twoapprox,lem:twoapprox-min-sum}.
    
     Moreover, we can compute $D_{\Pol}$ in $\mathcal{O}(n^4)$ time~\cite{Tan01a,TanJ17}. 
     The route $C_{ij}$ is a concatenation of $\partial\Pol[v_i,v_{j-1}]$ and the geodesic $\pi(v_i,v_{j-1})$, and can thus be computed in linear time for any $i,j\in 0,\dots,t-1$.
    The routes $D_{ij}$ can be computed in~$\mathcal{O}(n^3)$ time by modifying the algorithm of Jiang and Tan~\cite{TanJ17}. 
    Since there are $\mathcal{O}(t^2)$ such routes (all possible pairs of convex vertices), the total running time for computing all those routes~$W_{ij}$ is~$\mathcal{O}(t^2n^3)$.
\end{proof}

\section{Computational Complexity}\label{sec:complexity}
We now study the computational complexity of computing shortest segment watchman routes. 
We show that the problem is weakly \NP-hard under the min-max objective, even when restricted to simple polygons. 
For the min-sum objective, we demonstrate that the problem is \NP-hard in polygons with holes.

\subsection{Min-Max Segment Watchman Routes in Simple Polygons}\label{subsec:compl-min-max-segm}

We reduce from \textsc{Multiway Number Partitioning}~\cite{GareyJ79}. 
In particular, we consider the problem of partitioning a set of numbers into two subsets with equal sum. 
This problem, commonly referred to as \textsc{Partition}, is known to be weakly \NP-hard.

\begin{restatable}{theorem}{theoremMinMaxSegmentHardness}
    \label{thm:segment-hardness}
    The \textsc{Min-Max Segment Watchman Routes Problem} is weakly \NP-hard, even in simple polygons.
\end{restatable}

\begin{proof}
	We reduce from the \NP-complete \textsc{Partition Problem}~\cite{GareyJ79}. 
	For any given instance $\varphi = \{\alpha_1,\alpha_2,\dots,\alpha_n\}$ consisting of $n$ numbers, we construct a simple polygon $\Pol_\varphi$ with $2n+2$ vertices as an instance of the \textsc{Min-Max Segment Watchman Routes Problem} as follows; we refer to~\cref{fig:hardness-polygon} for an example.
	
	\begin{figure}[htb]
		\centering
		\includegraphics{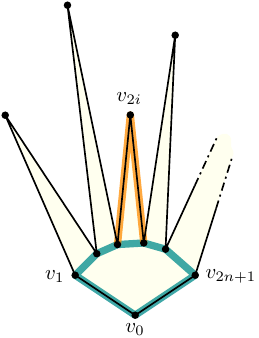}
		\caption{High-level idea of the type of polygon utilized in the \NP-hardness reduction.}
		\label{fig:hardness-polygon}
	\end{figure}
	
	The polygon $\Pol$ is star-shaped with the bottommost vertex $v_0$ in the kernel.
	Every $\alpha_i\in \varphi$ gives rise to a spike with convex vertex $v_{2i}$.
	For every $v_{2i}$, the corresponding spike's boundary consisting of the straight-line segments $\overline{v_{2i-1}v_{2i}}$ and $\overline{v_{2i}v_{2i+1}}$ has summed length $\alpha_i$; each segment having length $\alpha_i / 2$.
	We construct each spike so that its width (i.e., the distance between the vertices $v_{2i-1}$ and $v_{2i+1}$) is in $O(\varepsilon)$, where $\varepsilon$ denotes the length of the tour $v_0, v_1, v_3, \dots, v_{2i+1}, \dots, v_{2n+1}, v_0$, i.e., the tour starting and ending at $v_0$ and visiting all vertices of odd index in increasing~order.
	By scaling the instances, we can make $\varepsilon$ as small as needed.
	With this, we can guarantee that the ratio between any $\alpha_i$ and $\varepsilon$ is as large as~needed.
	
	Let the sum of the values from $\varphi$ equal $T$.
	It remains to argue that a partition of $\varphi$ into two sets of (nearly) equal sum (i.e., $T/2$) exists if and only if there are two watchman routes $W_1$, $W_2$ for $\Pol_\varphi$ such that every point within $\Pol_\varphi$ is segment-guarded, and the lengths of the routes are at most $T/2+\varepsilon$.
	
	\begin{claim}
		There exists a partition of $\varphi$ into two subsets $\varphi_1, \varphi_2$ of size $T/2$ each if there are two segment watchman routes $W_1, W_2$ of min-max length at most $T/2 + \varepsilon$.
	\end{claim}
	\begin{claimproof}
		Consider a solution to the min-max segment watchman route problem consisting of the tours $W_1$ and $W_2$, and assume that $W_1$ is the longer one, having length at most $T/2+\varepsilon$. 
		Clearly, the length of each spike's boundary corresponding to vertex $v_{2i}$ with $i\in [1,n]$ visited by $W_1$ adds up to $T/2$, as the ratio between $\varepsilon$ and each $\alpha_i$ is as large as needed.
		The total length of the polygons boundary is $T+\delta$ where $\delta$ denotes the length of the straight-line segments $\overline{v_0v_1}$ and $\overline{v_0v_{2n+1}}$ and $\delta\in O(\varepsilon)$.
		As every convex must be visited by at least one tour, all convex vertices $v_{2i}$ not visited by $W_1$ must be visited by $W_2$.
		However, the respective polygon boundary also adds up to $T/2$. 
		Therefore, $W_1$ and $W_2$ yield a solution to the partition instance.
	\end{claimproof}
	
	\begin{claim}
		There are two segment watchman routes $W_1, W_2$ of min-max length at most $T/2 + \varepsilon$, if there exists a partition of $\varphi$ into two subsets $\varphi_1, \varphi_2$ of size $T/2$ each.
	\end{claim}
	\begin{claimproof}
		Consider a partition of $\varphi$ into two subsets $\varphi_1, \varphi_2$ of size $T/2$. 
		Without loss of generality, we assign $\varphi_1$ to the first watchman as follows. 
		The watchman starts at the vertex $v_0$, and visits the convex vertex $v_{2i}$ corresponding to $\alpha_i$ for every $\alpha_i\in \varphi_1$ in increasing order of the index, finally returning back to $v_0$.
		Note that if $\alpha_1 \in \varphi_1$, the watchman also visits the convex vertex $v_1$; if $\alpha_{2n} \in \varphi_1$ it also visits the convex vertex $v_{2n+1}$.
		The subset $\varphi_2$ is analogously assigned to the second watchman.
		This yields a min-max segment watchman route of length at most $T/2+\varepsilon$.
	\end{claimproof}
	These two claims conclude the proof.
\end{proof}

\subsection{Min-Sum Segment Watchman Routes in Polygons with Holes}\label{subsec:compl-min-sum-kgon}

For polygons with holes, we lack conditions as those in~\cref{lem:conditions} for simple polygons. 
We overcome this in our \NP-hardness proof, which is based on a reduction from the \textsc{Euclidean Traveling Salesman Problem}, by constructing feasible solutions whose lengths sharply distinguish optimal from non-optimal TSP tours. 
As a result, computing an optimal segment watchman route necessitates solving the TSP instance.

\pagebreak
\begin{theorem}
    \label{thm:min-sum-hardness}
    The \textsc{Min-Sum Segment Watchman Routes Problem} is \NP-hard in polygons with holes.
\end{theorem}

\begin{proof}
    We reduce from the \textsc{Euclidean Traveling Salesman Problem}~\cite{DBLP:conf/stoc/GareyGJ76, p-etspi-77}.
    For a given set $\mathcal{P}$ of points in the plane, we construct a polygonal domain with a niche for each point on the boundary of the convex hull of $\mathcal{P}$ and a hole for each point in the interior of $\mathcal{P}$. 
    
    We may assume that for any pair of points, the distance between the two points is different from the distance between any other pair of points. 
    Now let $\Delta = \min_{p, q, p', q' \in \mathcal{P}} d(p, q) - d(p', q')$ be the smallest difference among all pairwise distances, and set $\delta \ll \Delta / | \mathcal{P}| $. With this, we can construct our polygonal domain, as depicted in~\Cref{fig:minsum-k-gon-hardness}. 

    \begin{figure}[htb]
        \centering
        \includegraphics[scale=0.85]{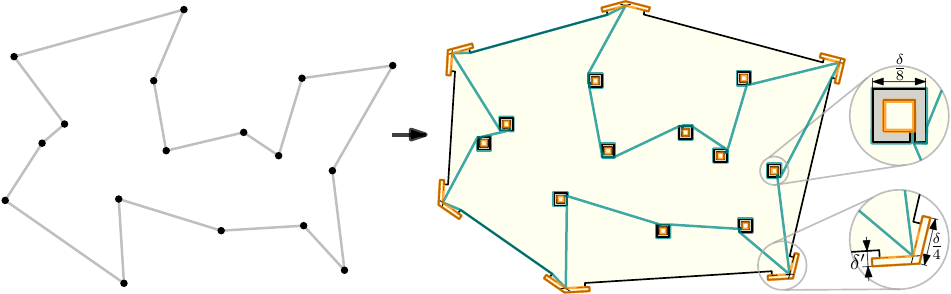}
        \caption{
        Schematic illustration of the reduction from a TSP instance to a polygonal domain. 
        The left figure shows a TSP tour on the input point set, while the right figure shows the corresponding routes in the constructed polygon. 
        Feasible watchman routes, including the detours around the convex hull of each hole (taken by both watchmen), are shown in green. For each niche and hole, $W_1$ additionally visits all convex vertices; these maximal detours are shown in orange.
        }
        \label{fig:minsum-k-gon-hardness}
    \end{figure}

    The outer polygon \Pol\ is a modification of the convex hull of $\mathcal{P}$, where each corner is replaced by a niche. Each such niche adds two essential cuts to \Pol\ that every route needs to visit, and the essential cuts lie outside the convex hull (of $\mathcal{P}$). Similarly, each polygonal hole is shaped such that it ``hides'' some part of the interior of \Pol\!, hence, adding an additional cut that every route needs to visit. 
    As every route of a segment watchman route is also a watchman route for the domain considered, every route needs to visit the essential cuts of the four niches and the ``door'' of every hole.

    Each hole $H$ is a polygon with 10 vertices. The convex hull of $H$ is a square with side length $\delta / 8$. In the inner part of the hole lies an inner square, which is connected to the outside of $H$ by a narrow tunnel, the \emph{door}. 
    See the top lens in \Cref{fig:minsum-k-gon-hardness} for an illustration. 
    To see the inner square of $H$, a watchman needs to walk through the corresponding door. 

    To place a niche at a corner of the outer polygon, eight additional vertices are necessary. For this, on each edge incident to a convex hull corner, a segment of length $\delta/4$ is cut off and parallel shifted outwards by $\delta' < \delta/8$. Then, to create the essential cut, the shifted corner is connected with the convex hull edges by three edgers of length $\delta'/2$ on each side, creating a chain of two convex vertices, the first of them is incident to the shifted corner edge, and two reflex vertices. See the bottom lens in \Cref{fig:minsum-k-gon-hardness}. 
    All edges are parallel or orthogonal to the convex hull edges. 
    The total length that is added to the outer boundary of the polygon is~$2 \delta/4 + 6 \delta'/2 < \delta$. 
    With this, each niche adds two essential cuts to the outer polygon which intersect in the interior of \Pol, but outside the convex hull. 

We obtain feasible watchman routes by adapting the TSP tour on $\mathcal{P}$ to visit the extensions in the holes and walking once around the convex hull of every hole. The latter adds a detour of $\leq \delta/2$. We denote these two routes as $W_1$ and $W_2$.
   Now, extend $W_1$ so that
    \begin{itemize}
        \item For every niche, $W_1$ visits all convex vertices of the niche, adding a detour~of~$\leq \delta$ 
        \item For every hole $H$, $W_1$ visits all vertices in the inner part of $H$, adding a detour of $\leq \delta/2$
    \end{itemize}
    
    We can prove that $W_1$ and $W_2$ are segment watchman routes in the polygonal domain. 
    First, consider a point $p$ that lies in the interior of one of the routes, say $W_1$. By construction, it is seen from $W_2$, i.e., there exist $w_2\in W_2$ that sees $p$. 
    Then the ray from $w_2$ in the direction of $p$ intersects $W_1$ at a point $w_1$ that sees $p$, and thus $p$ is segment-guarded by~$\overline{w_1w_2}$.

    Now consider a point $p$ that lies outside of each of the routes $W_1$ and $W_2$. Then $p$ lies inside a simple polygon $P'$ bounded by the two routes and the line $\ell$ through one of the edges of the outer polygon. 
    Let $\ell_p$ be the line parallel to $\ell$ that contains $p$. 
    The line $\ell$ either intersects two of the essential cuts of the outer polygon, or it intersects one of the essential cuts on one side of the edge extension and creates a pocket that contains a convex vertex on the other side of the edge extension. 
    Therefore, $\ell_p$ intersects both routes on one side of $p$ before it hits an edge of the polygonal domain, and at least one of the routes in the other direction. Hence, $p$ is segment-guarded. 

    Next, let $T_{OPT}$ be an optimal solution to the TSP in~$\mathcal{P}$, and let $T_2$ be the shortest solution to TSP in $\mathcal{P}$ that visits the points in a different order than $T_{OPT}$. 
    We argue that the length of each of the routes $W_1$ and $W_2$ is shorter than $T_2$.
    Each route $W_i$ takes a detour of at most~$\delta$ per corner niche and at most~$\delta$ per hole. With this, the following inequality holds:  
    \begin{equation*}
        |W_2| < |W_1| \leq  (T_{OPT} + |\mathcal{P}| \cdot \delta) 
        <  (T_{OPT} + \Delta) \leq  T_2.
    \end{equation*}
    Thus, contracting the loops around the holes and the detour at the niches yields a shortest route that visits all points in $\mathcal{P}$. 
    
    The routes $W_1$ and $W_2$ are feasible segment watchman routes, each of length at least~$T_{OPT}$, but strictly shorter than $T_2$. The optimal segment watchman routes also need to visit all essential cuts of the holes and niches, and have, hence, also length at least $T_{OPT}$. Moreover, they are at most as long as $W_1$ and $W_2$, hence, we need to solve the TSP on $\mathcal{P}$ to find the optimal segment watchman routes.
\end{proof}

\section{Generalization to \texorpdfstring{$\bm{k}$}{$k$}-Hull Watchman Routes}\label{sec:kgon-gen}

We now generalize our results on segment watchman to $k$-hull watchman routes. 
To this end, we prove sufficient conditions equivalent to those from~\cref{lem:conditions} for $k=2$ to all $k\ge 2$:
\begin{restatable}[The Conditions Lemma for $k$-Hull Watchman Routes]{lemma}{lemConditionsKGon}\label{lem:conditions-k}
    The routes $W_1, \dots, W_k$, $k \geq 2$, are $k$-hull watchman routes for \Pol\ if the following conditions hold:
    \begin{enumerate}[(i$_k$)]
        \item Every convex vertex is visited by one of $W_1, \dots, W_k$.\label{cond1-k}
        \item The visibility polygon of each convex vertex is visited by all of the routes $W_1, \dots, W_k$.\label{cond2-k}
        \item The routes $W_1, \dots, W_k$ are relatively convex.\label{cond3-k}
    \end{enumerate}
\end{restatable}

\begin{proof}
	By \Cref{clm:cond2}, Condition (\ref{cond2-k}$_k$) implies that $W_i$ is a watchman route for all $i\in\{1,\dots,k\}$. 
	
	In analogy to the proof of~\cref{lem:conditions}, we now show that Conditions (\ref{cond1-k}$_k$), (\ref{cond2-k}$_k$), and~(\ref{cond3-k}$_k$) imply that $W_1, W_2, \dots, W_k$ are $k$-hull watchman routes. 
	Consider a point $p\in P$.
	Because~$W_1, W_2, \dots, W_k$ are watchman routes, on each route, we have a point $w_i\in W_i$ that sees $p$. 
	We again define maximal wedges $F_i$ bounded by two rays starting at $p$, such that for every ray $\rho\in~F_i$ there is a point $w_i\in W_i$ in that direction that sees $p$. 
	Again, as $P$ is simple, all $F_i$ are single wedges.
	
	Each wedge $F_i$ covers either $360^\circ$ (if $p$ lies on or within the relatively convex route) or less than $180^\circ$ (because by Condition (\ref{cond3-k}$_k$) all routes are relatively convex).
	If at least one wedge $F_j$, $j\in\{1,\ldots,k\}$, covers $360^\circ$ around~$p$, then $p$ is $k$-hull-guarded: 
	$F_j$ covers $360^\circ$ around~$p$, and there exist $w_i\in W_i\ \forall i\in\{1,\ldots,k\}\setminus\{j\}$ that see $p$.
	These $w_i$ form a $(k-1)$-gon $Q^{(k-1)}$, which may or may not contain $p$. 
	If $Q^{(k-1)}$ contains $p$, we can pick any point on~$F_j$ that sees $p$ to extend $Q^{(k-1)}$ to a $k$-hull $Q^{k}$ that contains $p$. 
	(This is all we need,~$Q^{k}$ (and also~$Q^{(k-1)}$) may not be contained in $P$, but that is not a requirement for $k$-hull watchman routes.) 
	If $Q^{(k-1)}$ does not contain $p$, assume, without loss of generality, that all $w_i$ for $i\in\{1,\ldots,k\}\setminus\{j\}$ are contained in the right half plane with respect to a vertical line through $p$. 
	Let $w^u, w^\ell$ be the $w_i$ with maximum and minimum $y$-coordinate, respectively. 
	We shoot rays $\rho^u$ from $w^u$ and $\rho^\ell$ from $w^\ell$ through~$p$. 
	There exist points on $W_j$ in the left half plane lying strictly between $\rho^u$ and $\rho^\ell$ that see $p$. 
	Any such point will create the desired $k$-hull with~$Q^{(k-1)}$. 
	
	Hence, assume that none of the $F_i$ covers $360^\circ$ around $p$. Then an analogous argument to that in the proof of~\Cref{lem:conditions} yields a contradiction to Condition (\ref{cond3-k}$_k$): in addition to the~$F_i, i\in\{1,\ldots,k\}$, we have between one and $k$ maximal wedge(s), $F_\bullet^j, j\in\{1,\ldots,\bar{j}\}, \bar{j}\leq k$, bounded by two rays starting at $p$, such that for every ray $\rho$ in $F_\bullet^j$ there is no $w_i\in W_i$ for any $i$ in this direction that sees $p$. 
	By the same arguments as in the proof of~\cref{lem:conditions}, none of the $F_\bullet^j$ can cover more than $180^\circ$. Hence, we can pick points $w_i\in W_i$ in the $F_i$ such that the rays $\rho_i$ from $p$ towards $w_i$ are not fully contained in any half plane defined by a line through $p$. Hence, these $w_i$ form a $k$-hull $Q^{k}$ that contains $p$.
\end{proof}

\subsection{
Approximation Algorithms for \texorpdfstring{$\bm{k}$}{$k$}-Hull Watchman Routes}\label{sec:appx-kgon}
In order to generalize the algorithms from~\cref{sec:appx-segm}, we establish the following crucial lemma.

\begin{restatable}{lemma}{lemNonCrossingRoutes}\label{lem:non-crossing-routes}
Let $X_1$ and $X_2$ be subsets of $V_C$. Then $\H(X_1)$ and $\H(X_2)$ cross each other if and only if there exist vertices $v_a,v_b,v_c,v_d$ such that $v_a,v_c\in X_1$, $v_b,v_d\in X_2$, and $a<b<c<d$.
\end{restatable}

\begin{proof}
	Assume that there exists vertices $v_a,v_b,v_c,v_d$ such that $v_a,v_c\in X_1$, $v_b,v_d\in X_2$, and $a<b<c<d$. Let $v_a\in X_1$ be such that $a<b$ is the largest possible index, and let $v_c\in X_1$ be such that $c>b$ is the smallest possible index. Then $\H(X_1)$ contains the geodesic path~$\pi(v_a,v_c)$. This path splits \Pol\ into two polygons, one containing $v_b$ and the other containing $v_d$. Since $v_b,v_d\in X_2$, $\H(X_2)$ must cross $\pi(v_a,v_c)$.
	
	For the other direction, assume that $\H(X_1)$ and $\H(X_2)$ cross each other at a point $x$. Let $v_a\in X_1$ be the first vertex encountered when walking on $\H(X_1)$ from $x$ in one direction, and $v_c\in X_1$ the first vertex encountered when walking on $\H(X_1)$ from $x$ in the opposite direction. Then $x$ is on $\pi(v_a,v_c)$, which again splits \Pol\ into two polygons. Because $\H(X_2)$ is crossing $\H(X_1)$, each of these polygons contains at least one convex vertex in $X_2$, and therefore there exists $v_b,v_d\in X_2$ such that $a<b<c<d$.
\end{proof}

\medskip
Consider a subset $X\subseteq V_C$ of the convex vertices of \Pol, and by~\cref{obs:relative_hull_is_shortest}, let $C_X=\H(X)$ be a shortest route that visits the vertices in $X$.
Let $v_X$ be the vertex in $X$ with the smallest index, and denote by $D_X$ the shortest route that starts and ends at $v_X$, and that sees all the convex vertices in $V_C\setminus X$. 
Recall that $C_{\Pol}=\partial\Pol$ is a watchman route, $D_{\Pol}$ is the shortest floating watchman route in {\Pol}, and that $|D_{\Pol}|\le |C_{\Pol}|$.
We define $W_X\defeq \H(C_X \cup D_X)$ as the relative convex hull of the routes $C_X$ and $D_X$.

We construct our approximate solutions by iterating over all possible ways to partition the set $V_C$ of convex vertices into subsets $X_1,\dots,X_m$ for $m\le k$ such that the routes~$C_{X_1},\dots,C_{X_m}$ are disjoint. 
By \Cref{lem:non-crossing-routes}, two routes $C_{X_i}$ and $C_{X_j}$ are disjoint if and only if $X_i$ and $X_j$ are ``non-crossing'' with respect to the order along $\partial\Pol$.
We say that two sets $X,Y\subseteq V_C$ are \emph{non-crossing} if there exist two indices $0\le a,b\le t-1$ such that (i) the indices of vertices in $X$ are in $[a,b-1]$ and (ii) the indices of vertices in $Y$ are in $[b,a-1]$.
We say that~$k$ subsets $X_1,\dots,X_k\subseteq V_C$ are a \emph{non-crossing partition} of $V_C$ if for a convex~$|V_C|$-gon the relative convex hulls of the vertices corresponding to the subsets $X_i$ do not intersect.
It~is well-known that the number of non-crossing partitions of size~$k$ for a set of $t$ points is equal to the Narayana number $N(k,t)= \frac{1}{t} \binom{t}{k} \binom{t}{k-1}$~\cite{narayana-numbers-oeis} with $\frac{1}{t} \binom{t}{k} \binom{t}{k-1}=\mathcal{O}(t^{2k-2})$.

\subparagraph*{Min-Max Variant.}

Let $\mathcal{S}^m$ be the set of all possible non-crossing partitions of size $m$ for the set $V_C$. For the approximated solution, we iterate over all the non-crossing partitions and pick the best option, as follows.
For every length $2\le m\le k$, set:
\begin{equation}\label{set-routes-min-max}
(W_1,W_2,\dots,W_m)
=
\argmin_{\{X_1,\dots,X_m\}\in\mathcal{S}^m}\max\{\vert W_{X_1}\vert,\vert W_{X_2}\vert,\dots,\vert W_{X_m}\vert\},
\end{equation}
and set $W_{m+1}=\dots=W_k=D_{\Pol}$.
Then, we take the optimal set over all choices of $m$. 
Finally, if $\max_i |W_i|>|C_{\Pol}|$, we pick $W_1=C_{\Pol}$ and $W_2=W_3=\dots=W_k=D_{\Pol}$. 

By \Cref{lem:conditions-k,obs:RCH-is-valid}, $(W_1,W_2,\dots,W_k)$ is a feasible solution for the $k$-hull watchman routes problem.
If $\OPT_k(\Pol)$ denotes the size of an optimal solution for the min-max $k$-hull watchman routes problem in~\Pol, we obtain the following lemma:
\begin{restatable}{lemma}{lemMaxRouteSmallerKOPT}\label{lem:twoapprox-k}
$\max\big\{\vert W_{1}\vert,\vert W_{2}\vert,\dots,\vert W_{k}\vert\big\} \leq k \cdot \OPT_k(\Pol)$.
\end{restatable}

\begin{proof}
	Throughout this proof, when referring to a route $W$ considered by our algorithm, we assume that it has both a corresponding $C$-part and $D$-part, however, either may be empty.
	
	Let $W_1^*,W_2^*,\dots, W_k^*$ be an optimal solution for the $k$-hull watchmen routes problem, i.e., ${\max\big\{\vert W_1^*\vert,\vert W_2^*\vert,\dots,\vert W_k^*\vert\big\}=\OPT_k(\Pol)}$. 
	Without loss of generality, we may assume that $W_1^*,W_2^*,\dots, W_k^*$ are as short as possible.
	
	For a route $W_j^*$, denote by $V_C(W_j^*)$ the set of convex vertices of \Pol\ that it visits.
	Notice that if there exists a route $W_j^*$ with $V_C(W_j^*)=V_C$, then $|W_j^*|\ge |C_{\Pol}|$, and it is enough to pick all the other routes to be $D_{\Pol}$. Therefore in this case $W_1=C_{\Pol}$ and $W_2=W_3=\dots=W_k=D_{\Pol}$ is an optimal solution to the problem, and the theorem holds.
	
	Consider the intersection graph $G$ of $W_1^*,W_2^*,\dots, W_k^*$: the vertices of $G$ are the routes $W_1^*,W_2^*,\dots, W_k^*$, and there is an edge between two vertices if the corresponding routes intersect. 
	Consider a maximal connected component of $G$, and assume w.l.o.g. that its corresponding routes are $W_1^*,W_2^*,\dots, W_r^*$, for some $r\le k$.
	Consider the ``merged'' route $W'=H(W_1^*\cup \dots\cup W_r^*)$, then $|W'|\le r\cdot\max \{|W_1^*|,\dots,|W_r^*| \} \le r\cdot\OPT_k(\Pol)$.
	If there is a single connected component in $G$, then $W'$ visits all the convex vertices of \Pol, and therefore $|C_{\Pol}|\le |W'|\le k\cdot\OPT_k(\Pol)$. Because $|D_{\Pol}|\le |W'|$, we get that $W_1=C_{\Pol}$, and $W_2=W_3=\dots=W_k=D_{\Pol}$ is a $k$-approximation, as required.
	
	Else, assume that there is more than one connected component. By ``merging'' each connected component of $G$ into a single route in the same way, we obtain a set of non-intersecting routes that together visit all the convex vertices of \Pol\ (and each of them sees all the convex vertices).
	Let $W'_1,\dots, W'_m$ be the merged routes, then the routes $(W'_1,\dots, W'_m,W'_{m+1},\dots,W'_k)$, where $W'_i=D_{\Pol}$ for $m+1\le i\le k$, make a feasible solution. Moreover, because the size $r$ of each component is strictly smaller than $k$, we have $\max_{i=1\dots k} |W'_i|\le (k-1)\cdot\OPT_k(\Pol)$. Now, the routes $(W'_1,\dots,W'_m)$ are pairwise non-crossing.
	Consider the sets $V_C(W'_1),\dots,V_C(W'_m)$. By the minimality of $W^*_i$, the route $H(V_C(W'_i))$ is contained in $W'_i$ for every $1\le i\le m$, and thus the routes $H(V_C(W'_1)),\dots,H(V_C(W'_m))$ are also non-crossing. By \Cref{lem:non-crossing-routes}, the sets $V_C(W'_1),\dots,V_C(W'_m)$ form a non-crossing partition of $V_C$. 
	Denote $X_i=V_C(W'_i)$ for $1\le i\le m$, then
	$|C_{X_i}|\le |W'_i|\le (k-1)\cdot\OPT_k(\Pol)$ because $C_{X_i}$ is the shortest route that visits $X_i$.
	Because $D_{X_i}$ is the shortest route that sees only the convex vertices in $X_i$ and visits only a single convex vertex, it must be shorter than the route in the optimal solution that visits that vertex. Therefore, we have $|D_{X_i}|\le \OPT_k(\Pol)$ for any $X_i$, and we obtain that \[\max_{i=1\dots k}|W_{X_i}|\le \max_{i=1\dots k}\{|C_{X_i}|+|D_{X_i}|\}\le \max_{i=1\dots k}\{|C_{X_i}|\}+\OPT_k(\Pol)\le  k\cdot\OPT_k(\Pol).\]
	Because $X_1,\dots,X_m$ is a non-crossing partition, it was considered when choosing $W_1,\dots W_k$. Thus $\max_{i=1\dots k}|W_i|\le \max_{i=1\dots k}|W_{X_i}|$, and the claim follows.
\end{proof}

\medskip
We thus derive the following theorem. 
Its proof mirrors that of \Cref{thm:twoapprox}: feasibility follows from \Cref{lem:conditions-k}, and the approximation guarantee from \Cref{lem:twoapprox-k}. 
Moreover, the number of partitions is $\mathcal{O}(t^{2k-2})$.
\begin{theorem}
    Given a simple polygon \Pol\ with $n$ vertices, out of which $t$ are convex, one can compute a $k$-approximation for the $k$-hull watchman route problem under the min-max criterion in $\mathcal{O}(n^4+t^{2k-2}\cdot n^3)$~time.
\end{theorem}

\subparagraph*{Min-Sum Variant.}
We now generalize our approximation algorithm for the min-sum objective. 
Here, instead of setting the routes as in \Cref{set-routes-min-max}, we set
\begin{equation}
(W_{1},W_{2},\dots,W_m)
=\!\!\!\!\!\!\!
\argmin_{\{X_1,\dots,X_m\}\in\mathcal{S}^m}\sum_{i=1}^m \big\vert W_{X_i}\big\vert.
\end{equation}

We also define $D_{X_i}$ slightly differently as the shortest floating watchman route that sees all convex vertices in $V_C\setminus X_i$, without the need to visit any convex vertex.
Let $\OPT^\Sigma_k(\Pol)$~be the size of an optimal solution for the \textsc{Min-Sum $k$-hull-Watchman Routes Problem}~in~\Pol. 
\begin{restatable}{lemma}{lemTwoApproxKRoutes}\label{lem:twoapprox-min-sum-k}
$\sum_{i=1}^k\vert W_{i}\vert \leq 2 \cdot {\OPT^\Sigma_k}(\Pol)$.
\end{restatable}

\begin{proof}
	Let $W_1^*,W_2^*,\dots, W_k^*$ be an optimal solution, i.e., $\sum_{i=1}^k\vert W_i^*\vert=\OPT^\Sigma_k(\Pol)$.
	First notice that because $D_\Pol$ is the shortest floating watchman route, we have $|D_{\Pol}|\le W^*_i$ for every $i\in\{1,\dots,k\}$, and thus $|D_{\Pol}|\le {\OPT^\Sigma_k}(\Pol)/k$.
	
	Consider a maximal connected component of the graph $G$ as in the proof of \Cref{lem:twoapprox-k}, and assume, without loss of generality, that its corresponding routes are $W_1^*,W_2^*,\dots, W_r^*$, for some $r\le k$.
	Consider the ``merged'' route $W'=H(W_1^*\cup \dots\cup W_r^*)$, then $|W'|\le \sum_{i=1}^r|W_i^*|$.
	If there is a single connected component in $G$, then $W'$ visits all the convex vertices of \Pol, and therefore $|C_{\Pol}|\le |W'|$. 
	We thus have
	\[
	\sum_{i=1}^k |W_i|\le |C_\Pol|+(k-1)\cdot|D_\Pol|\le \OPT^\Sigma_k(\Pol)+(k-1)\cdot {\OPT^\Sigma_k}(\Pol)/k< 2\cdot{\OPT^\Sigma_k}(\Pol).
	\]
	
	Else, assume that there is more then one connected component. Let $W'_1,\dots, W'_m$ be the merged routes that correspond to connected components of $G$, as in the proof of \Cref{lem:twoapprox-k}.  
	Each such route $W'_i$ is the relative convex hull of a set of optimal routes from a different connected component of $G$.
	Because each route from $W_1^*,W_2^*,\dots, W_k^*$ participates in exactly one of the merged routes $W'_1,\dots, W'_m$, we have $\sum_{i=1}^m |W'_i|\le \sum_{i=1}^k |W^*_i|$. 
	
	Let $X_i=V_C(W'_i)$ for $i\in\{1,\dots m\}$ be the sets of convex vertices that are visited by $W'_i$, and notice that $X_1,\dots,X_m$ form a non-crossing partition of $V_C$. 
	Set $X_{m+1}=\dots=X_k=\emptyset$, so by definition $C_{X_{m+1}}=\dots=C_{X_k}=\emptyset$ and $D_{X_{m+1}}=\dots=D_{X_k}=D_\Pol$.
	Consider the routes~$W_{X_i}=\H(C_{X_i}\cup D_{X_i})$ for $i\in\{1,\dots,k\}$.
	Because $C_{X_i}$ is the shortest route that visits~$X_i$, we have $|C_{X_i}|\le|W'_i|$ for $i\in\{1,\dots,m\}$, and thus $\sum_{i=1}^m |C_{X_i}|\le\sum_{i=1}^m |W'_i| \le\sum_{i=1}^k |W^*_i|$. 
	We~therefore have    
	\begin{align*}
		\sum_{i=1}^k|W_{X_i}|&\le \sum_{i=1}^k(|C_{X_i}|+|D_{X_i}|)=\sum_{i=1}^m|C_{X_i}|+\sum_{i=m+1}^k|C_{X_i}|+\sum_{i=1}^k|D_{X_i}|\\
		&\le \sum_{i=1}^m|C_{X_i}|+\sum_{i=1}^k|D_\Pol|\le
		\sum_{i=1}^k|W^*_i|+k\cdot|D_\Pol|\le 2\cdot\OPT^\Sigma_k(\Pol).
	\end{align*}
	
	Because $X_1,\dots,X_m$ is a non-crossing partition, it was considered when choosing $W_1,\dots W_k$. Thus $\sum_{i=1}^k|W_i|\le \sum_{i=1}^k|W_{X_i}|$, and the claim follows.
\end{proof}

We obtain the following theorem. 
Again, our proof is similar to our proof of \Cref{thm:twoapprox}: using \Cref{lem:conditions-k} for feasibility and \Cref{lem:twoapprox-min-sum-k} for the approximation factor, and where the number of partitions is $\mathcal{O}(t^{2k-2})$. 
However, the running time increases slightly, as $D_{X}$ is now defined to be a floating watchmen route, whose computation requires $\mathcal{O}(n^4)$ time.
\begin{theorem}
    Given a simple polygon \Pol\ with $n$ vertices, out of which $t$ are convex, one can compute a $2$-approximation for the $k$-hull watchmen route problem under the min-sum criterion in $\mathcal{O}(t^{2k-2}\cdot n^4)$ time.
\end{theorem}

\subsection{\NP-Hardness Results for \texorpdfstring{$\bm{k}$}{$k$}-Hull Watchman Routes}

By leveraging similar arguments as in the proof of~\cref{thm:segment-hardness}, we straightforwardly obtain \NP-hardness for the min-max variant of our problem with $k$ watchmen by reducing from \textsc{Multiway Number Partitioning} instead of \textsc{Partition}.

\begin{corollary}\label{cor:k-gon-min-max-hardness}
    The \textsc{Min-Max $k$-hull Watchman Routes Problem} is \NP-hard, even in simple polygons.
\end{corollary}

Moreover, we extend the proof of~\cref{thm:min-sum-hardness} by constructing $k$ routes, each of which visits every hole and every corner niche, and traverses the convex hull of each hole. 
Additionally, for every corner niche and every hole, at least one of these $k$ routes visits all convex vertices. 
This straightforwardly establishes:
\begin{corollary}
    \label{thm:min-sum-hardness-k}
    The \textsc{Min-Sum $k$-hull Watchman Routes Problem} is \NP-hard in polygons with holes for $k \geq 2$.
\end{corollary}

\section{Conclusions and Future Work}
In this paper, we studied segment watchman routes in simple polygons.
First, we identified sufficient conditions for two routes to constitute valid segment watchman routes. 
Leveraging these conditions, we then designed a $2$-approximation algorithm for segment watchman routes for both the min-max and min-sum optimization objectives.

Moreover, we establish both the 
\NP-hardness of the \textsc{Min-Max Segment Watchman Routes Problem} even in simple polygons, as well as the \NP-hardness of the \textsc{Min-Sum $k$-hull Watchman Routes Problem} in polygons with holes, highlighting the inherent computational difficulty of the problem.

Finally, we extended all our results to $k$-hull watchman routes, obtaining  a $k$- and a $2$-approximation algorithm for $k$-hull watchman routes for the min-max and min-sum optimization objectives, respectively.

\subparagraph*{Interior \texorpdfstring{$\bm{k}$}{$k$}-Hulls.}
In the problems considered in this paper, the $k$-hull (for all $k>2$) defined by the points on the watchman routes that see a point $p\in P$ is not required to be contained in~$P$. 
One can naturally define variants of the two problems in which this additional containment condition is imposed. 
However, the conditions of~\cref{lem:conditions-k} are not sufficient in this setting, see~\cref{fig:no-int-tri}: the three routes satisfy all conditions of the lemma.
\begin{figure}[h]
    \centering
    \includegraphics[page=1, scale=0.9]{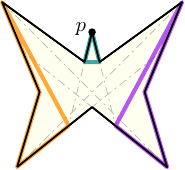}
    \caption{An example for the conditions from~\cref{lem:conditions-k} not being sufficient for interior-$k$-hull watchman routes.}
    \label{fig:no-int-tri}
\end{figure}
The point $p$ is seen from exactly one point on each of the orange and violet routes. 
As $p$ is a convex vertex, the corresponding point on the green route must be $p$ itself. 
Since the triangle determined by these three points intersects the boundary of \Pol, it is not fully contained in the polygon.
Finding sufficient conditions for this variant, as well as developing approximation algorithms, remains an open problem.

\bibliography{bibliography}

\end{document}